\documentclass[12pt]{article}
\usepackage{graphicx}          
\usepackage{units}
                               
\usepackage{amsfonts,amsmath,amssymb}
\usepackage{mathrsfs}
\usepackage{appendix}
\newtheorem{theorem}{Theorem}

\newtheorem{lemma}[theorem]{Lemma}

\newtheorem{proposition}[theorem]{Proposition}
\newtheorem{remark}[theorem]{Remark}

\newenvironment{proof}[1][Proof]{\noindent\textbf{#1.} }{\ }
\newenvironment{keywords}[1][Keywords]{\noindent\textbf{#1:} }{}
\newenvironment{dedication}{\begin{quotation}\begin{center}\begin{em}}{\par\end{em}\end{center}\end{quotation}}

\renewcommand\appendix{\par
\setcounter{section}{0}%
\setcounter{subsection}{0}%
\setcounter{table}{0}
\setcounter{figure}{0}
\gdef\thetable{\Alph{table}}
\gdef\thefigure{\Alph{figure}}
\section*{Appendix}
\gdef\thesection{\Alph{section}}
\setcounter{section}{1}}

\begin{document}

\title{Quantum Filtering for Multiple Input Multiple Output Systems Driven by Arbitrary Zero-Mean Jointly Gaussian Input Fields\thanks{Research supported by the Australian Research Council}}
\author{Hendra~I.~Nurdin 
\thanks{Hendra I. Nurdin is with the School of Electrical Engineering and Telecommunications, UNSW Australia, 
Sydney NSW 2052, Australia. Email: h.nurdin@unsw.edu.au}}
\date{}

\maketitle \thispagestyle{empty}

\begin{dedication}
Dedicated to the memory of Slava Belavkin  
\end{dedication}

\begin{abstract}
In this paper, we treat the quantum filtering problem for multiple input multiple output (MIMO) Markovian open quantum systems coupled to multiple boson fields in an arbitrary zero-mean jointly Gaussian state, using the reference probability approach formulated by Bouten and van Handel as a quantum version of a well-known method of the same name from  classical nonlinear filtering theory, and exploiting the generalized Araki-Woods representation of Gough. This includes Gaussian field states such as vacuum, squeezed vacuum, thermal, and squeezed thermal states as special cases.  The contribution is a derivation of the general quantum filtering equation  (or stochastic master equation as they are known in the quantum optics community) in the full MIMO setup for any zero-mean jointy Gaussian input field states, up to some mild rank assumptions on certain matrices relating to the measurement vector.
\end{abstract}


\section{Introduction}
\label{sec:intro}

Quantum filtering theory was developed starting from the late 70s by the pioneering efforts of Viacheslav ``Slava'' Belavkin, documented in a long sequence of highly original papers, see, e.g., \cite{VPB79,VPB83,VPB92a,VPB92}\footnote{For a complete list, see Belavkin's memorial homepage at the University of Nottingham, https://www.maths.nottingham.ac.uk/personal/vpb/}. The remarkable ideas developed therein, that extend key concepts from stochastic filtering and control theory for classical (i.e., non-quantum) Markovian systems to their quantum counterparts, were well ahead of their time. Indeed, they would not be implementable until the beginning of the 21st century as the technology for experimental quantum optics   advanced enough to make continuous monitoring of Markovian quantum optical systems possible. Thus, Belavkin's quantum filtering equation, the quantum analog of the Kushner-Stratonovich equation from classical nonlinear filtering theory, was independently discovered by several physicists working in quantum optics \cite{DCM92,DZR92,CARM93,HC93}, most notably through the work of Howard Carmichael within a framework known as quantum trajectory theory \cite{HC93}. In the terminology of the latter theory, Belavkin's filtering equation is known as the stochastic master equation.

The quantum filtering equations have been derived for Markovian systems coupled to various continuous-mode boson fields in specific Gaussian states (a precise definition of Gaussian states of the field will be given in Section \ref{sec:Gaussian-states}), including the vacuum state \cite{VPB92a}, squeezed vacuum state \cite{LB04}, and coherent states \cite{GK10}, typically under the measurement of only a single field. More recently, they have even been extended to highly non-Gaussian field states such as single photon states \cite{GJN13,GJNC12}, multi-photon states \cite{SZX13}, superposition of continuous-mode coherent states \cite{GJNC12}, and a class of continuous matrix product states \cite{GJN14}, using Markovian \cite{GJNC12,GJN14}, and non-Markovian embedding techniques \cite{GJN13,SZX13}. However, despite these advances, to the best of the author's knowledge, there is not yet a complete treatment of quantum filtering for systems driven by  multiple fields in  an arbitrary zero-mean jointly Gaussian field state and under arbitrary linear measurements performed on multiple outputs of the system. Here we address the problem for any zero-mean jointly Gaussian field states with arbitrary second order correlations. The quantum filtering equation has previously been derived for a system coupled to vacuum fields under so-called ``dyne measurements'' of multiple output fields\footnote{After the completion of this work, the author also became aware of the work \cite{GMS14} that treats the special case of dyne measurements with  thermal state inputs using quantum trajectory theory.} using quantum trajectory theory \cite[Section 4.5.2]{WM10}, with an application to Markovian feedback of MIMO open Markovian systems developed in \cite{CW11}.  Provided that certain mild rank conditions on certain matrices related to the measurement vector are fufilled in certain scenarios, our results give the most general form of the quantum filtering equation for multiple input multiple output (MIMO) systems coupled to multiple fields in  any zero-mean jointly Gaussian state.

The organization of this paper is as follows. In Section \ref{sec:prelim}, we give a brief review of Gaussian states of boson fields, the quantum stochastic calculus of Hudson-Parthasarathy, and the reference probability approach to quantum filtering. In Section \ref{sec:main}, we derive the main results of this paper. We begin this section by illustrating the calculations involved in the simplest case of a system coupled to a single vacuum boson field under an arbitratry linear measurement of the output field. The calculations are then extended to a system coupled to multiple vacuum fields under arbitrary linear measurements on multiple output fields. The latter results are then applied to obtain the quantum filtering equation for a system coupled to multiple boson fields in an arbitrary zero-mean jointly Gaussian state. Finally, Section \ref{sec:conclu} gives the conclusion of the paper.

\section{Preliminaries}
\label{sec:prelim}

\textbf{Notation.} $\imath=\sqrt{-1}$. If $X=[X_{jk}]$ is a matrix of Hilbert space operators or complex numbers, then $X^*$ is the adjoint of $X$, $X^{\top}=[X_{kj}]$, and $X^{\#}=[X_{jk}^*]$. $L^2([0,\infty);\mathbb{C}^n)$ denotes the Hilbert space of $\mathbb{C}^n$-valued square integrable functions on $[0,\infty)$.  Unless specified otherwise, all vectors are asumed to be represented as column vectors. $\mathscr{B}(\mathcal{H})$ denotes the class of all bounded operators on the Hilbert space $\mathcal{H}$, $\Gamma_s(\mathcal{H})$ denotes the symmetric boson Fock space over a Hilbert space $\mathcal{H}$, and $\mathcal{F}_n=\Gamma_s(L^2([0,\infty);\mathbb{C}^n))$. $\langle \cdot \rangle$ denotes quantum expectation, ${\rm Tr}(\cdot)$ denotes the trace of a trace-class operator, $\delta_{jk}$ is the Kronecker delta, and $\delta(\cdot)$ denotes  the Dirac delta function. If $X$ is a matrix of trace-class operators then ${\rm Tr}(X)=[{\rm Tr}(X_{jk})]$. If $X,Y$ are column vectors of Hilbert space operators then $[X,Y]$ denotes a matrix whose $j,k$-th entry is $[X_j,Y_k]$. $I_n$ denotes the $n \times n$ identity matrix, and $0_{m \times n}$ denotes an $m \times n$ zero matrix, however, subscripts may be dropped when the dimensions can be identified unambiguously from the context.

\subsection{Gaussian states of the field and their Fock space representations}
\label{sec:Gaussian-states}
We will consider Markovian open quantum systems that are coupled to $n$ continuous-mode boson fields indexed by $j=1,2,\ldots,n$ with annihilation field operators $b_j(t)$ satisfying the field commutation relations $[b_j(t),b_k(t')^*]=\delta_{jk} \delta(t-t')$ and $[b_j(t),b_k(t')]=0$. Let us introduce the shorthand notation,
$$
\breve{b}=\left[\begin{array}{c} b \\ b^{\#} \end{array}\right].
$$
Following the treatment in \cite{GJN10}, a zero-mean Gaussian state $\omega_{N,M}(\cdot) =\langle \cdot \rangle$ of $n$ continuous-mode boson fields can be described by the correlation function
\begin{eqnarray*}
\langle b_{j}^{\ast }(t)b_{k}(t^{\prime })\rangle =N_{jk}\;\delta
(t-t^{\prime }),\\
\langle b_{j}(t)b_{k}(t^{\prime })\rangle =M_{jk}\;\delta
(t-t^{\prime }).
\end{eqnarray*}
That is, 
\begin{equation}
\left\langle \breve{b}\left( t\right) \breve{b}^{*}\left( t^{\prime
}\right) \right\rangle \equiv F\delta \left( t-t^{\prime }\right) ,
\end{equation}
where  $F$ is an $2n \times 2n$ Hermitian matrix of the form
\begin{equation}
F=\left[
\begin{array}{cc}
I+N^{\top } & M \\ 
M^{* } & N
\end{array}
\right] 
\label{eq:F}
\end{equation}
with
$N=[N_{jk}]=N^{*}$ and $M=[M_{jk}]=M^{\top }$. By the definition of a Gaussian state, clearly $F \geq 0$ (for a Hermitian matrix, $\geq 0$ denotes that the matrix is positive semidefinite) and this entails that  $N\geq 0$ .  For the $n=1$ case, $N$ and $M$ are scalars and the positivity condition is easily seen to be $N\geq 0$ with $|M|^{2}\leq N\left( N+1\right) $. 

An important special case of a Gaussian state is the vacuum state, when $N=0$ and $M=0$. This is the state when the fields are empty (devoid of any photons). The vacuum state for the field is characterized by 
\begin{eqnarray*}
\left\langle \exp \left( \imath \int_{0}^{\infty }\breve{u}\left( t\right)^{*} \breve{%
b}\left( t\right) dt \right) \right\rangle _{\text{vac}}= \exp \left( -\frac{1}{2}%
\int_{0}^{\infty }\breve{u}\left( t\right)^{*} F_{\text{vac}}\breve{u}%
\left( t\right) dt \right),
\end{eqnarray*}
for any $u \in L^2([0,\infty);\mathbb{C}^n)$, where the superscript ${\rm vac}$ denotes vacuum. 

For convenience, we will work with so-called smeared versions of the singular field operators, namely $B_j(f)= \int_{0}^{\infty} f(s)^* b_j(s)ds$  for any $f \in L^2([0,\infty);\mathbb{C})$ and its adjoint process $B^*_j(f)=\int_{0}^{\infty} f(s) b_j(s)^*ds$ (which we will often write as $B_j(f)^*$ for notational convenience) as they are well-defined  and more regular mathematical objects, and can be manipulated using the quantum stochastic calculus of Hudson and Parthasarathy \cite{HP84,KRP92,Mey95}. They satisfy the canonical commutation relations $[B_j(f),B_k^*(g)]=\delta_{jk}\int_{0}^{\infty} f(s)^*g(s) ds$, and the concrete realization of the processes $B(f)=(B_1(f),B_2(f),\ldots,B_n(f))^{\top}$ and $B(f)^{\#}=(B_1^*(f),B_2^*(f),\ldots,B_n^*(f))^{\top}$ on a suitable Hilbert space are dependent on the state of the field. 
However, for arbitrary Gaussian states one can relate the associated realization of $B(f)$ and $B^*(f)$ to the vacuum state representation of these operators, via the so-called generalized Araki-Woods representation \cite{Gough03,HHKKR02,GJN10}. In particular, any smeared operator $B(f)$ associated with a zero-mean Gaussian state $\omega_{N,M}$, with $F \geq 0$ as given in (\ref{eq:F}), has a Fock space representation of the form
\begin{equation}
B(f) =  C_1 A_1(f) +  C_2 A_2(f) +  C_3 A_2(f)^{\#}, \label{eq:GAW}
\end{equation}
for some appropriate complex $n \times n$ matrices $C_1$,  $C_2$ and $C_3$  that are determined by the values of the parameters $N$ and $M$ of $\omega_{N,M}$, see \cite{Gough03,GJN10} for how to construct these matrices, and where $A_1$ and $A_2$ are two independent vacuum smeared annihilation processes that can each be realized on a distinct copy  of the Fock space $\mathcal{F}_n$. 
Note that $C_1,C_2,C_3$ cannot be arbitrary, but they must be such that 
the commutation relations $[B_j(f),B_k^*(g)]=\delta_{jk} \int_{0}^t f(s)^*g(s)$ hold for any $f,g \in L^2([0,\infty);\mathbb{C})$.

\subsection{Quantum stochastic calculus}
 \label{sec:QSC}
For the special case of a joint vacuum state of the fields, let us introduce the integrated field annihilation process $A_j(t)=A_j(1_{[0,t]})$ ($1_{[0,t]}$ denoting the indicator function on the set $[0,t]$) and its adjoint process, the integrated field creation process, $A_j^*(t)=A_j^*(1_{[0,t]})$. In the vacuum representation, their future-pointing It\={o} increments 
$dA_j(t)=A_j(t+dt)-A_j(t)$ and $dA_j^*(t)=A_j^*(t+dt)-A_j^*(t)$  
satisfy the quantum It\={o} table
\begin{equation*}
\begin{tabular}{l|ll}
$\times $ & $dA_{k}^{*}$ & $dA_{k}$ \\ \hline
$dA_{j}$ & $\delta_{jk} dt$ & 0 \\ 
$dA_{j}^{*}$ & 0 & 0
\end{tabular}
\end{equation*}
We may also define the counting process  (or gauge process)
\begin{equation*}
\Lambda _{jk}(t)=\int_{0}^{t}b_{j}^{\ast }(r)b_{k}(r)dr,
\end{equation*}
which may be included in the It\={o} table \cite{HP84}. The additional
non-trivial products of differentials are
\begin{equation*}
d\Lambda _{jk}dA_{l}^{*}=\delta _{kl}dA_{j}^{*},dA_{j}d\Lambda
_{kl}=\delta _{jk}dA_{l},d\Lambda _{jk}A\Lambda _{li}=\delta _{kl}d\Lambda
_{ji}.
\end{equation*}
Using the processes $A=(A_1,A_2,\ldots,A_n)^{\top}$, $A^{\#}=(A_1^*,A_2^*,\ldots,A_n^*)^{\top}$ and $\Lambda=[\Lambda_{jk}]$, one may define  quantum stochastic integrals of adapted processes on the tensor product of the system and joint Fock space of the fields. The system is the quantum mechanical object that is being coupled to the fields, and adapted means that at time $t$ the process acts trivially on the portion of the boson Fock space after time $t$ , see, e.g.,  \cite{HP84,KRP92,Mey95} for details. An adapted process commutes at time $t$ with all of the future pointing differentials.
The product of two adapted processes $X(t)$ and  $Y(t)$ is again adapted, and the increment of the product obeys the quantum It\={o} rule
$$
d(X(t)Y(t)) = (dX(t))Y(t) + X(t) dY(t) + dX(t) dY(t),
$$
Based on these quantum stochastic integrals, one may define quantum stochastic differential equations (QSDEs). An important QSDE that describes the joint unitary evolution of an open Markovian process coupled to  vacuum boson fields, common in quantum optics and related fields, is the Hudson-Parthasarathy QSDE given by
\begin{equation}
dU(t) = (-(\imath H +\nicefrac{1}{2}L^*L)dt  + dA(t)^*L - L^{*} SdA(t)+{\rm Tr}((S-I)d\Lambda(t)^{\top}))U(t), \label{eq:HP-QSDE-vac}
\end{equation}
with initial condition $U(0)=I$. The input field $A(t)$ after the interaction becomes transformed into the output field $A^{\rm out}(t) =U(t)^* A(t) U(t)$. Here $H$ represents the Hamiltonian of the system that is coupled to the field, $L$ is a column vector of coupling operators that models the coupling of the system to the annihilation process $A$, while $S$ is a unitary operator that represents the coupling of the system to the process $\Lambda$ of the field.   This QSDE has a unique solution whenever $S,L,H$ are bounded operators. Moreover, in that case the solution is guaranteed to be unitary. Since we are interested in the form of the filtering equation, to keep the exposition as concise as possible and technicalities to a minimum, we assume throughout that $S,L,H$ are bounded operators.

The non-vacuum It\={o} table can be directly constructed by exploiting the generalized Araki-Woods representation (\ref{eq:GAW}) and  the vacuum It\={o} table. Recall that $A_1(f)$ and $A_2(f)$ in (\ref{eq:GAW}) are vacuum representations on distinct copies of the Fock space $\mathcal{F}_n$. The extended It\={o} table for the integrated operators $B_j(t)=B_j(1_{[0.t]})$ and $B_j^*(t)=B_j^*(1_{[0.t]})$ when the field is in an arbitrary Gaussian state with $F$ as given in (\ref{eq:F}) is then
\begin{equation}
\begin{tabular}{l|ll}
$\times $ & $dB_{k}^{*}$ & $dB_{k}$ \\ \hline
$dB_{j}$ & $(\delta_{jk} +N_{kj}) dt$ & $M_{jk}dt$ \\ 
$dB_{j}^{*}$ & $M_{kj}^\ast dt $ & $ N_{jk} dt$.
\end{tabular} \label{tbl:non-vac}
\end{equation}
Note that in general Gaussian states the counting process $\Lambda$ need not be defined. We can also define a QSDE of the Hudson-Parthasarathy type but in which the vacuum field operators $A$ and $A^{\#}$ are replaced by field operators $B$ and $B^{\#}$  corresponding to a non-vacuum zero-mean jointly Gaussian state of the field. This yields the QSDE (without the counting process $\Lambda$),
\begin{eqnarray}
dU(t) &=&  \bigl(-(\imath H +\nicefrac{1}{2} (L^*(I+N^{\top})L^{\#}+L^{\top}NL^{\#}-L^*ML^{\#}-L^{\top}M^{\#}L))dt \notag \\
&&\qquad + dB(t)^*L -L^* dB(t) \bigr)U(t) \label{eq:HP-QSDE-nonvac1}
\end{eqnarray}
with initial condition $U(0)=I$. As with the vacuum case, after interaction with the system, $B$ is transformed to $B^{\rm out}$ according to $B^{\rm out}(t)=U(t)^* B(t) U(t)$. Using the generalised Araki-Woods representation (\ref{eq:GAW}), we can express the QSDE in terms of the vacuum operator $A(t)=[\begin{array}{cc} A_1(t)^{\top} & A_2(t)^{\top} \end{array}]^{\top}$,
\begin{eqnarray}
dU(t) &=&  (-(\imath H +\nicefrac{1}{2}L_{N,M}^*L_{N,M})dt + dA(t)^*L_{N,M} -L_{N,M}^*dA(t))U(t), \label{eq:HP-QSDE-nonvac2}
\end{eqnarray} 
with 
$$
L_{N,M}=\left[\begin{array}{cc} C_1^* L \\ C_2^*L-C_3^{\top}L^{\#} \end{array}\right].
$$ 
 
 \subsection{Reference probability approach to quantum filtering}

Suppose that the system of interest lives on a Hilbert space $\mathfrak{h}_{\rm sys}$ and has initial state $\omega_{\rm sys}(\cdot) ={\rm Tr}(\rho_{\rm sys} \cdot)$ for some density operator $\rho_{\rm sys}$ in $\mathscr{B}(\mathfrak{h}_{\rm sys})$. The system is coupled to multiple fields in a joint vacuum state so that the joint initial state of the system and fields is $\varpi(\dot)=\omega_{\rm sys} \otimes \langle \Omega| \cdot |\Omega \rangle$, with $|\Omega \rangle$ denoting the vacuum state of the fields. Thus,
\begin{eqnarray*}
\varpi(X) = {\rm Tr}( \rho_{\rm sys}\otimes |\Omega\rangle \langle \Omega | X), 
\end{eqnarray*}
for any operator $X$ in $\mathscr{B}(\mathfrak{h}_{\rm sys}) \otimes \mathscr{B}(\mathcal{F}_n)$.  We now consider the scenario where we have an arbitrary linear  measurement $Y(t)$ of a quadrature of the output field $A^{\rm out}$. That is, $Y(t)$ is an $m \times 1$ vector that is a linear combination of $A^{\rm out}(t)$ and $A^{\rm out}(t)^{\#}$,
\begin{equation}
Y(t) = G^{\#} A^{\rm out}(t) + G A^{\rm out}(t)^{\#},  \label{eq:Y-meas}
\end{equation}
with $G \in \mathbb{C}^{m \times n}$ and $m \leq n$. We require that  $[\begin{array}{cc} G^{\#} & G \end{array}]$ is full rank and satisfies
\begin{equation}
\left[ \begin{array}{cc} G^{\#} & G \end{array}\right] \mathbb{K}_n \left[ \begin{array}{c} G^{*} \\ G^{\top} \end{array}\right]=0, \label{eq:G-commutation}
\end{equation}
with
$$
\mathbb{K}_n =\left[\begin{array}{cc}  0 & I_n \\ -I_n & 0 \end{array}  
\right].
$$
Note that the full rank requirement on $[\begin{array}{cc} G^{\#} & G \end{array}]$ entails no loss of generality since  $m \leq n$ and $[\begin{array}{cc} G^{\#} & G \end{array}]$ not being full rank implies that there is redundancy (linear dependence) in the measurement vector $Y$ that can be removed to reduce $[\begin{array}{cc} G^{\#} & G \end{array}]$ to the full rank case. We now comment on the following {\em subtle} point. Observe that it is not obvious that $G$ is guaranteed to be full rank even if $[\begin{array}{cc} G^{\#} & G\end{array}]$ does have this property, unless $G$ is a real matrix, $G^{\#}=G$, or $m=1$. 
\begin{remark}
\label{rem:pathology} From this point onwards, we enforce the assumption that $G$ is full rank. Since  $[\begin{array}{cc} G^{\#} & G \end{array}]$ is already assumed to be full rank, it seems reasonable to expect that, at least generically, $G$ would also be full rank. The full rankness of $G$ will play a crucial role in the  derivations in later sections.
\end{remark}

Let $j_t(X)=U(t)^* X U(t)$ denote the evolution of $X$ in the Heisenberg picture for any $X \in \mathscr{B}(\mathfrak{h}_{\rm sys})$. Then we have
\begin{lemma}
\label{lem:non-demo-X}
Let $G$ satisfy (\ref{eq:G-commutation}). Then the measurement $Y(t)$ satisfies $[Y(t),Y(s)]=0$ for all $s,t \geq 0$ and  $[j_t(X),Y(s)]=0$ for any system operator $X$ and any $0 \leq s \leq t$.
\end{lemma}
\begin{proof}
Recall that $[A(t),A(s)]=0=[A^{\rm out}(t),A^{\rm out}(s)]$, $[A(t),A(s)^{\#}]=\min(t,s)I_n =[A^{\rm out}(t),A^{\rm out}(s)^{\#}]$ for all $s,t \geq 0$. Using these properties, direct calculation then shows that the condition on $G$ implies that $[Y(t),Y(s)]=0$ for all $0 \leq s \leq t$. Also, recall that $[j_t(X),A^{\rm out}(s)]=0=[j_t(X),A^{\rm out}(s)^{\#}]$ for all $0 \leq s \leq t$ by the cocycle property of the solution of the Hudson-Parthasarathy QSDE. Using these properties and the definition of $Y(t)$, direct calculation  shows that  $[j_t(X),Y(s)]=0$ for all $0 \leq s \leq t$.  
\hfill $\Box$
\end{proof} 

For any system operator $X$, the quantum filtering problem is to find an optimal mean-square estimate of $j_t(X)$ based on the observation of $Y$ up to time $t$: $\{Y(s);\, 0\leq s \leq t\}$. Let $\mathscr{Y}_t$ be the commutative von Neumann algebra generated by (the spectral projections of) the elements of $\{Y(s);\, 0\leq s \leq t\}$. To obtain an optimal mean-square estimate of $j_t(X)$, our goal is to derive a stochastic differential equation, the quantum filtering equation, for the quantity
$
\pi_t(X) = \varpi(j_t(X) | \mathscr{Y}_t),
$  
where $\varpi(j_t(X) | \mathscr{Y}_t)$ denotes the quantum conditional expectation of $j_t(X)$ on $\mathscr{Y}_t$ with respect to the state $\varpi$. This is a well-defined quantity since $[j_t(X),Y(s)]=0$ for all $0 \leq s \leq t$, and furnishes an optimal mean-square estimate of $X(t)$ given $\mathscr{Y}_t$.

Introduce a process $Z$ in an analogous way to $Y$ as  
\begin{equation}
Z(t) =  G^{\#} A(t) + G A(t)^{\#},  \label{eq:Z-meas-vac}
\end{equation}
where $A^{\rm out}$ in the definition of $Y$ has been replaced by $A$. Then by property (\ref{eq:G-commutation}) we also have that $[Z(t),Z(s)]=0$ for all $0 \leq s \leq t$. Denote the commutative von Neumann algebra generated by (the spectral projections of) the elements of $\{Z(s);\, 0\leq s \leq t\}$ by $\mathscr{Z}_t$. We are now ready to explain how to derive quantum filtering equations using the quantum reference probability approach introduced by Bouten and van Handel \cite{BvH06,BvH08} as a quantum adaptation of the reference probability approach from classical nonlinear stochastic filtering theory. In the quantum context, the basis for this approach is the following theorem  \cite{BvHJ07}

\begin{theorem}
Let $U(t)$ be the unitary defined by the QSDE (\ref{eq:HP-QSDE-vac}) and let $Q_t(\cdot)$ be a time-dependent state on the joint system and field Fock space defined by
$$
Q_t(A)= \varpi(U(t)^* A U(t)),
$$
for all $A \in \mathscr{B}(\mathfrak{h}_{\rm sys}) \otimes \mathscr{B}(\mathcal{F}_n)$ then
$$
\pi_t(X)=U(t)^* Q_t(X | \mathscr{Z}_{t}) U(t).  
$$
Moreover, if there is an adapted process $V(t)$ which is the solution of a QSDE such that $V(t) \in \mathscr{Z}_{t}$ for all $t$ and satisfies the identity
$$
Q_t(X) = \varpi(V(t)^*X V(t))
$$
for all $X$ in $\mathscr{B}(\mathfrak{h}_{\rm sys})$ and all $t \geq 0$, then
$$
\pi_t(X)=\frac{\sigma_t(X)}{\sigma_t(I)},.
$$
where
$\sigma_t(X) = U(t)^*\varpi(V(t)^* X V(t) | \mathscr{Z}_{t})U(t)$ all $X$ in $\mathscr{B}(\mathfrak{h}_{\rm sys})$.
\end{theorem}

The crux of the approach is to construct a process $V(t)$ that satisfies the conditions of the theorem. We will do this for the quantum filtering problem of interest in this paper in the next section. Since $\sigma_t(X), j_t(X)\in \mathscr{Y}_t$, they are isomorphic to  classical stochastic processes on a common probability space, and it is possible to write down a classical stochastic differential equation (SDE) for the increment of $\sigma_t(X)$ and $\pi_t(X)$. The SDE for $\sigma_t(X)$ is known as the quantum Zakai equation while that for $\sigma_t(X)$ is called the quantum Kushner-Stratonovich equation, and are quantum counterparts of the SDEs of the same name appearing in classical stochastic filtering theory. The quantum Kushner-Stratonovich equation is also known as the Belavkin master equation, being first obtained by Slava Belavkin.

\section{Main results}
\label{sec:main} 
In this section, using the reference probability approach described in the previous section, we will derive the quantum filtering equation for systems driven by fields in an arbitary zero-mean jointly Gaussian state with arbitrary linear measurements performed on its outputs. However, to fix the main ideas and simplify the subsequent exposition, we first review the simple case of a system driven by a single input field in the vacuum state with an arbitrary linear measurement performed on its output field. Then the result is extended to systems driven by multiple vacuum fields with arbitrary linear measurements at its output before finally being applied to systems driven by fields in any zero-mean jointly Gaussian state.

\subsection{Quantum filtering of a system coupled to a single vacuum field with an arbitrary linear measurement on its output}

Consider the Hudson-Parthasarathy QSDE driven by a single field in a vacuum state,
\begin{equation}
dU(t) = (-\imath (H+\nicefrac{1}{2}L^*L) dt +L dA(t)^* -L^* SdA(t) + (S-I)d\Lambda(t))U(t),\; U(0)=I, \label{eq:HP-QSDE-single}
\end{equation}
and let $A^{\rm out}(t)=U(t)^*A(t)U(t)=j_t(L)+j_t(S)dA(t)$. We have a measurement of the form,
\begin{equation}
Y(t) = g^*A^{\rm out}(t) + g A^{\rm out}(t)^*, \label{eq:Y-meas-single-vac}
\end{equation}
for some complex number $g \neq 0$, and thus $Z(t)=g^* A(t) + g A(t)^*$.  We have the following proposition.
\begin{proposition}
\label{pro:single-V} Let $V(t)$ be an adapted process defined as the solution to the QSDE
$$
dV(t) = \left(-\bigl(\imath H+\nicefrac{1}{2}L^*L \bigr)dt + (L/g) dZ(t)\right)V(t),
$$
with $V(0)=I$.
Then
$$
\varpi(U(t)^* X U(t)) = \varpi(V(t)^* X V(t))
$$
for all $X \in \mathscr{B}(\mathfrak{h}_{\rm sys})$ and all $t \geq 0$.
\end{proposition}
\begin{proof}
Let $|\eta \rangle$ be a pure state on $\mathfrak{h}_{\rm sys}$. Since $A(t) |\Omega \rangle =0$ and $\Lambda(t)| \Omega \rangle =0$, we have that $U(t) |\eta \Omega \rangle=V(t)|\eta \Omega \rangle$ for all $|\eta \rangle \in \mathfrak{h}_{\rm sys}$ by a trick that is attributed in \cite{BvH06} to Holevo, see, e.g., \cite[Lemma 6.2]{BvH06}, for a proof. It then follows by inspection that $U(t) \rho_{\rm sys} \otimes |\Omega \rangle \langle \Omega |U(t)^* = V(t) \rho_{\rm sys} \otimes |\Omega \rangle \langle \Omega |V(t)^*$ for all initial system density operators $\rho_{\rm sys}$, and therefore
$ 
\varpi(U(t)^* X U(t)) = \varpi(V(t)^* X V(t))
$ 
for all $X \in \mathscr{B}(\mathfrak{h}_{\rm sys})$. \hfill $\Box$
\end{proof}

Introduce the Linbladian superoperator  $\mathcal{L}_{H,L}$ as
$$
\mathcal{L}_{H,L}(X)=-\imath [X,H] +L^*XL-\nicefrac{1}{2}(L^*LX+XL^*L),
$$
for any $X \in \mathscr{B}(\mathfrak{h}_{\rm sys})$. By a straightforward application of the quantum It\={o} rules we have the following theorem.
\begin{theorem}
\label{thm:filter-single-in} The quantum Zakai equation for the system (\ref{eq:HP-QSDE-single}) with measurement $Y$ given by (\ref{eq:Y-meas-single-vac}) is
\begin{eqnarray*}
d\sigma_t(X) =\sigma_t(\mathcal{L}_{H,L}(X))dt+\sigma_t(XL/g+L^*X/g^*)dY(t),
\end{eqnarray*}
and the quantum Kushner-Stratonovich equation for the conditional expectation $\pi_t(X)$ is
\begin{eqnarray*}
d\pi_t(X) &=& \pi_t(\mathcal{L}_{H,L}(X))dt +\bigl(\pi_t(g^*XL +gL^*X) -\pi_t(X)\pi_t(g^*L +g L^*)\bigr)|g|^{-2} d\nu(t),
\end{eqnarray*}
where $\nu$ is the innovations process and is a $\mathscr{Y}_t$-martingale defined by
$$
\nu(t)=\int_{0}^t (dY(s)-\pi_s(g^*L+gL^*)ds).
$$
Moreover, $\nu$ is a Wiener process.
\end{theorem}
\begin{proof}
First note that by the quantum It\={o} rule (see \cite{BvH06,BvH08} for a justification of the manipulations involved)
\begin{eqnarray*}
\varpi(V(t)^*XV(t)\bigr| \mathscr{Z}_t)&=& \int_{0}^t   \varpi(dV(t)^*XV(t) + V(t)^* X dV(t) + dV(t)^*XdV(t)\bigr| \mathscr{Z}_{t}),\\
&=& \varpi\left( \int_{0}^t(dV(t)^*XV(t) + V(t)^* X dV(t) + dV(t)^*XdV(t) \bigr| \mathscr{Z}_{t}\right),\\
&=& \varpi\left( \int_{0}^t(V(t)^*\mathcal{L}_{H,L}(X)V(t) +V(t)^*( XL/g+L^*X/g^*)V(t) dZ_{t} \biggr| \mathscr{Z}_{t}\right),\\
&=& \int_{0}^t \left(\varpi(V(s)^*\mathcal{L}_{H,L}(X)V(s) | \mathscr{Z}_{s}\right.)ds \\
&&\quad \left. + \varpi \bigl(V(s)^*\bigl( XL/g+L^*X/g^*\bigr)V(s) \bigr| \mathscr{Z}_{s} \bigr)dZ_{s} \right),
\end{eqnarray*}
so that
\begin{eqnarray*}
d\varpi(V(t)^*XV(t)| \mathscr{Z}_t) &=&  \varpi(V(t)^*\mathcal{L}_{H,L}(X)V(t) \mid \mathscr{Z}_{t})dt \\
&&\quad + \varpi \left(V(t)^*( XL/g +L^*X/g^*)V(t) | \mathscr{Z}_{t}\right)dZ_{t}.
\end{eqnarray*}
Using the above expression for $d\varpi(V(t)^*XV(t)\mid \mathscr{Z}_t)$ together with the quantum It\={o} rule, and the fact that $\sigma_t(X)$ being in $\mathscr{Y}_{t}$ commutes with $j_t(X) \in \mathscr{Y}_{t}$ for any system operator $X$, yields
\begin{eqnarray*}
d\sigma_t(X) &=& d(U(t)^* \varpi(V(t)^*XV(t)\mid \mathscr{Z}_t) U(t)),\\
&=& (dU(t)^*) \varpi(V(t)^*XV(t) \mid \mathscr{Z}_t) U(t))+U(t)^*(d \varpi(V(t)^*XV(t))\mid \mathscr{Z}_t) U(t)) \\
&&\quad + U(t)^* \varpi(V(t)^*XV(t)\mid \mathscr{Z}_t) dU(t) + dU(t)^*( d\varpi(V(t)^*XV(t))\mid \mathscr{Z}_t) U(t) \\
&&\quad + dU(t)^* \varpi(V(t)^*XV(t)\mid \mathscr{Z}_t) dU(t)+ U(t)^*d \varpi(V(t)^*XV(t)\mid \mathscr{Z}_t) dU(t) \\
&&\quad+ dU(t)^* d\varpi(V(t)^*XV(t)\mid \mathscr{Z}_t) dU(t),\\
&=&( \sigma_t(\mathcal{L}(X)) + g j_t(L^*)\sigma_t(XL/g+L^*X/g) + g^* j_t(L)\sigma_t(XL/g+L^*X/g)  )dt \\
&&\qquad + gj_t(S^*)\sigma_t(XL/g+L^*X/g)dA(t)^* +g^*j_t(S)\sigma_t(XL/g+L^*X/g)dA(t), \\
&=& \sigma_t(\mathcal{L}(X))dt+\sigma_t(XL/g+L^*X/g^*)dY(t).
\end{eqnarray*}
Note in particular that the quantum Zakai equation has no terms involving $d\Lambda(t)$ as they vanish in the calculations. 

With the quantum Zakai equation in hand it is a straightforward but tedious task to compute the quantum Kushner-Stratonovich QSDE again by straightforward application of the classical It\={o} rule. Recall that since $\sigma_t(X)$ and $\sigma_t(I)$ are processes in the commutative von Neumann algebra $\mathscr{Y}_t$ and are isomorphic to two classical stochastic processes that can be realized on the same classical probability space, the differential of the quotient can be computed with the classical It\={o} rule,
\begin{eqnarray*}
d\pi_t(X) &=& d\left(\frac{\sigma_t(X)}{\sigma_t(I)}\right),\\
&=& \frac{d\sigma_t(X)}{\sigma_t(I)}+\sigma_t(X)\left(-\frac{d\sigma_t(I)}{\sigma_t(I)^2} +\frac{d\sigma_t(I)^2}{\sigma_t(I)^3}\right)-d\sigma_t(X)\left(-\frac{d\sigma_t(I)}{\sigma_t(I)^2} +\frac{d\sigma_t(I)^2}{\sigma_t(I)^3}\right),\\
&=& \left(\pi_t(\mathcal{L}_{H,L}(X)) -|g|^2 \pi_t(XL/g+L^*X/g^*)\pi_t(L/g+L^*/g^*) \right. \\
&&\quad \left. +|g|^2 \pi_t(X)\pi_t(L/g+L^*/g^*)^2 \right)dt\\
&&\quad +(\pi_t(XL/g+L^*X/g^*)-\pi_t(X)\pi_t(L/g+L^*/g^*)dY(t),\\
&=&  \pi_t(\mathcal{L}_{H,L}(X))+(\pi_t(g^*XL+L^*X^*g)-\pi_t(X)\pi_t(g^*L+L^*g))|g|^{-2}d\nu(t),
\end{eqnarray*}
where $\nu$ is the innovations process of the filter as defined in the theorem. That $\nu(t)$ is a $\mathscr{Y}_t$-martingale and a Wiener process follows analogously from the proof of \cite[Theorem 7.1]{BvHJ07}.
\hfill $\Box$
\end{proof}

\subsection{Quantum filtering of a system coupled to a multiple vacuum fields with arbitrary linear measurements on its outputs}
\label{sec:filtering-multi-vac}
In this section, we turn to deriving the quantum filtering equation for a MIMO system driven by multiple fields in a joint vacuum state with arbitrary simultaneous linear measurements performed on multiple outputs of the system. The results of this section can then be applied immediately  to address systems driven by multiple fields in an arbitrary zero-mean jointly Gaussian state. The Hudson-Parthasarathy QSDE is as given in (\ref{eq:HP-QSDE-vac}) and the output fields are elements of the vector $A^{\rm out}(t)=U(t)^*A(t)U(t)$. The measurement will be an $m  \times 1$ vector of the form (\ref{eq:Y-meas}) with $G$ satisfying (\ref{eq:G-commutation}). 
The following lemma will be useful in the subsequent development.
\begin{lemma}
\label{lem:Xi-construct} 
If $[\begin{array}{cc} G^{\#} & G \end{array}]$ is full rank and $G$ satisfies (\ref{eq:G-commutation}) then if $m < n$
there exists a matrix $H \in \mathbb{C}^{(n-m) \times n}$ such that the matrix $W=[\begin{array}{cc} G^{\top} & H^{\top}\end{array}]^{\top}$
satisfies 
\begin{equation}
[\begin{array}{cc} W^{\#} & W \end{array}] \mathbb{K}_n \left[\begin{array}{c} W^{*} \\ W^{\top} \end{array}\right]=0, \label{eq:W-commutation}
\end{equation}
and $[\begin{array}{cc} W^{\#} & W \end{array}]$ is full rank.
\end{lemma}
\begin{proof}
Let $Q^{\rm out} (t) = A^{\rm out}(t)  + A^{\rm out}(t)^{\#}$ and $P^{\rm out}(t)= -\imath A^{\rm out}(t)  + \imath A^{\rm out}(t)^{\#}$ be the amplitude and phase quadratures of $A^{\rm out}(t)$, respectively.
Then we can write $Y(t) = (G + G^{\#}) Q^{\rm out}(t) + (-\imath G+\imath G^{\#})P^{\rm out}(t)$.  
Since $G$ satisfies (\ref{eq:G-commutation}) direct calculation verifies that
$$
[\begin{array}{cc} G + G^{\#} & -\imath G+\imath G^{\#} \end{array}]
$$ 
satisfies
$$
[\begin{array}{cc} G + G^{\#} & -\imath G+\imath G^{\#} \end{array}]\mathbb{K}_n \left [\begin{array}{c} (G + G^{\#})^{\top} \\ (-\imath G+\imath G^{\#})^{\top} \end{array}\right]=0.
$$
Now, since $[\begin{array}{cc} G^{\#} & G \end{array}]$ is full rank, it follows from the construction employed in the proof of \cite[Lemma 6]{Nurd11} (see also the proof of \cite[Lemma 6]{WNZJ14}) that one can construct $H \in \mathbb{C}^{(n-m) \times n}$ such that  
$$
\left[\begin{array}{cc} G + G^{\#} & -\imath G+\imath G^{\#} \\  H + H^{\#}  & -\imath H+\imath H^{\#} \end{array}\right] \mathbb{K}_n \left [\begin{array}{cc} (G + G^{\#})^{\top} & (H + H^{\#})^{\top} \\ (-\imath G+\imath G^{\#})^{\top} & (-\imath H +\imath H^{\#})^{\top} \end{array}\right]=0,
$$
with 
$$
\left[\begin{array}{cc} G + G^{\#} & -\imath G+\imath G^{\#} \\  H + H^{\#}  & -\imath H+\imath H^{\#} \end{array}\right]
$$
full rank. It follows immediately from this that the matrix $[\begin{array}{cc} W^{\#} & W \end{array}]$ satisfies (\ref{eq:W-commutation}) and is also full rank by multiplying the matrix above on the right by the invertible matrix,
$$
\frac{1}{2}\left[ \begin{array}{cc} I_n & I_n  \\ -\imath I_n & \imath I_n \end{array} \right].
$$
\hfill $\Box$
\end{proof} 

\begin{remark}
As in Remark \ref{rem:pathology} for $G$, we shall also need to enforce a separate assumption that the $n \times n$ matrix $W=[\begin{array}{cc} G^{\top} & H^{\top}\end{array}]^{\top}$ is full rank, hence invertible.
\end{remark}
     
Using the above lemma, the following proposition may be proved in the same manner as Proposition \ref{pro:single-V}.
\begin{proposition}
\label{prop:multi-eq-unitary} Let $U(t)$ be the unitary solution of the QSDE (\ref{eq:HP-QSDE-vac}) and $Z$ be as given in (\ref{eq:Z-meas-vac}). Let $H$ and $W$ be as in Lemma \ref{lem:Xi-construct}, $L_W = W^{-\top}L$ if $m<n$ and $L_W=G^{-\top}L$ if $m=n$,  and
$Z_W(t) = W^{\#} A(t) + WA(t)^{\#}$ (with $W=G$ if $m=n$). Also, let $V(t)$ be an adapted process defined as the solution to the QSDE
$$
dV(t) = \left(-\bigl(\imath H+\nicefrac{1}{2}L^*L \bigr)dt +  L_W^{\top} d Z(t)\right)V(t),
$$
with $V(0)=I$. Then $[Z_W(t),Z(s)]=0$ for all $s,t \geq 0$ and
$$
\varpi(U(t)^* X U(t)) = \varpi(V(t)^* X V(t))
$$
for all $X \in \mathscr{B}(\mathfrak{h}_{\rm sys})$ and all $t \geq 0$.
\end{proposition}
\begin{proof}
The proof follows {\em mutatis mutandis} from the proof of Proposition \ref{pro:single-V} with the following sequence of substitutions in the QSDE (justified as before by the fact that $|\Omega \rangle$ is a joint vacuum state of the fields),
\begin{eqnarray*}
L^{\top}dA(t)^{\#}-L^*dA(t) &\rightarrow&L^{\top}dA(t)^{\#},\\
&\rightarrow&  L^{\top} W^{-1} \left[\begin{array}{c} G \\ H \end{array}\right]dA(t)^{\#},\\
&\rightarrow&   L^{\top}W^{-1} \left[\begin{array}{c} G^{\#} \\ H^{\#} \end{array}\right]dA(t) + L^{\top}W^{-1} \left[\begin{array}{c} G \\ H \end{array}\right]dA(t)^{\#},\\
&\rightarrow&  L^{\top} W^{-1}dZ_{W}(t).
\end{eqnarray*}
Finally, note that  $[Z_W(t),Z_W(s)]=0$ for all $s,t\geq 0$ can be shown along the lines of the proof of Lemma \ref{lem:non-demo-X}. Moreover, we note that by definition, $Z$ is a subvector of $Z_W$. Hence, it must be that  $[Z_W(t),Z(s)]=0$ for all $s,t \geq 0$. \hfill $\Box$
\end{proof}

With the above proposition we obtain the analog of Theorem \ref{thm:filter-single-in} for the MIMO case, proved by performing similar calculations.
\begin{theorem}
\label{thm:filter-multi-in} Let $\varpi(Z_W(t) | \mathscr{Z}_t)=K_W Z(t)$ with $K_W \in \mathbb{R}^{n \times m}$. The quantum Zakai equation for the system (\ref{eq:HP-QSDE-vac}) with measurement $Y$ given by (\ref{eq:Y-meas}) is
\begin{eqnarray}
d\sigma_t(X) = \sigma_t(\mathcal{L}_{H,L}(X)) dt + \sigma_t(L_W^*X+XL_W^{\top})K_W dY(t), \label{eq:Zakai-multi-vac}
\end{eqnarray}
and the quantum Kushner-Stratonovich equation for the conditional expectation $\pi_t(X)$ is
\begin{eqnarray}
d\pi_t(X) &=& \pi_t(\mathcal{L}_{H,L}(X))dt + (\pi_t(L_W^*X+XL_W^{\top}) -\pi_t(X)\pi_t(L_W^*+L_W^{\top})) K_W d\nu(t) \label{eq:KS-multi-vac}
\end{eqnarray}
where $\nu$ is the innovations process and is a $\mathscr{Y}_t$-martingale defined by
$$
\nu(t)=\int_{0}^t (dY(s)-G^{\#}G^{\top}K_W^{\top} \pi_s(L_W^{\#}+L_W)ds). 
$$  
Moreover, $\nu$ is a Wiener process.
\end{theorem}
\begin{remark}
If $X$ is a matrix of compatible operators then $\sigma_t(X)=[\sigma_t(X_{jk})]$. Similarly, $\pi_t(X)=[\pi_t(X_{jk})]$.
\end{remark}
\begin{proof}
By a similar calculation to the proof of Theorem \ref{thm:filter-single-in}, we have that
\begin{eqnarray*}
\varpi(V(t)^*XV(t) \mid \mathscr{Z}_t)&=& \int_{0}^t   \varpi\bigl(dV(s)^*XV(s) + V(t)^* X dV(s) + dV(s)^*XdV(s) \bigr| \mathscr{Z}_s\bigr),\\
&=& \int_{0}^t  \varpi(V(s)^*\mathcal{L}(X)V(s) \mid \mathscr{Z}_{s})ds  \\
&&\quad +  \int_{0}^t \varpi \bigl(V(s)^*\bigl(L_W^*X + XL_W^{\top}\bigr)V(s) d Z_W(s) \mid  \mathscr{Z}_s \bigr),
\end{eqnarray*}
where $L_W$ is as defined in Proposition \ref{prop:multi-eq-unitary}. Now, recall from the proposition that the components of $Z$ commute with one another and with those in $Z_W$ at any two arbitrary times $s,t \geq 0$. Moreover, $Z$ is a subvector of $Z_W$. Therefore, $Z$ and $Z_W$ are isomorphic to two correlated classical vector Wiener processes and hence  there is an optimal  estimate $\hat{Z}(t) = \varpi(Z_W(t)| \mathscr{Z}_t)$ of $Z_W(t)$ given $\mathscr{Z}_t$. Note further, however, that being Wiener processes, $Z_W$ and $Z$ are $\mathscr{Z}_t$-martingales, with $Z$ being a subvector of $Z_W$, implying that $\varpi(Z_W(t)| \mathscr{Z}_t)=K_W Z(t)$
for some constant real $n \times m$ matrix $K_W$ (with the upper $m \times m $ block of $K_W$ obviously being $I_{m \times m}$), and one can write $Z_W(t) = \hat{Z}(t) +E(t) = K_WZ(t) +E(t)$, where $E(t)$ is another  vector Wiener process that is independent of $Z(t)$. Using this decomposition it follows that \cite{BvH06}
\begin{eqnarray*}
\varpi(V(t)^*XV(t) \mid \mathscr{Z}_t) &=& \int_{0}^t  \varpi(V(s)^*\mathcal{L}(X)V(s) \mid \mathscr{Z}_{s})ds  \\
&&\quad +  \int_{0}^t \varpi(V(s)^*(L_W^* X+ XL_W^{\top}\bigr)V(s)\bigr|\mathscr{Z}_s ) K_W  d Z_{s} ,
\end{eqnarray*}
Using the above expression for $d\varpi(V(t)^*XV(t)|\mathscr{Z}_t)$ together with the quantum It\={o} rule, yields, by  similar calculations to that in the proof of Theorem \ref{thm:filter-single-in}, the quantum Zakai equation (\ref{eq:Zakai-multi-vac}) and the quantum Kushner-Stratonovich equation (\ref{eq:KS-multi-vac}).  \hfill $\Box$
\end{proof}

\begin{remark}
A few remarks are now in order.
\begin{enumerate}
\item Once $W$ has been constructed, $K_W$ can be easily computed by the conditional expectation formula for jointly Gaussian random variables (since $Z_W(t)$ and $Z(t)$ are jointly Gaussian at any time $t \geq 0$).

\item Note that by (\ref{eq:G-commutation}), we have that $GG^*=G^{\#}G^{\top}$ (i.e., $GG^*$ is real and symmetric) so that the term $G^{\#}G^{\top}K_W^{\top}$ appearing under the stochastic integral for $\nu$ is real since $K_W$ is real.

\item When $m=n$ we have that $K_W=I$. Using the fact that $GG^*=G^{\#}G^{\top} \Leftrightarrow G=G^{\#}G^{\top}G^{-*}$, $\nu$ takes the form
$$
\nu(t)=\int_{0}^t (dY(s)- \pi_s(G^{\#}L+ GL^{\#})ds), 
$$ 
as would be expected.
\end{enumerate}
\end{remark}

We can write $\pi_t(X)={\rm Tr}(\hat{\rho}(t)X)$ for some stochastic system density operator $\hat{\rho}(t) \in \mathscr{Y}_t$ and all $X \in \mathscr{B}(\mathfrak{h}_{\rm sys})$. From the  SDE for $\pi_t(X)$, we can immediately obtain the density operator-valued SDE for $\hat{\rho(t)}$
as:
\begin{equation}
d\hat{\rho}(t) = \mathcal{L}_{H,L}^{\star}(\hat{\rho}(t))dt +\bigl( \hat{\rho}(t) L_W^* +L_W^{\top}\hat{\rho}(t)-{\rm Tr}(\hat{\rho}(t) (L_W^*+L_W^{\top}))\hat{\rho}(t)\bigr)K_W d\nu(t), \label{eq:SME-multi-vac}
\end{equation}
where $\mathcal{L}_{H,L}^{\star}$ is the Liouvillian superoperator
\begin{equation}
\mathcal{L}_{H,L}^{\star}(\rho) = \imath [\rho,H]+L\rho L^*-\nicefrac{1}{2}(L^*L \rho+\rho L^*L). \label{eq:ME-multi-vac}
\end{equation}
It is the SDE for $\hat{\rho}(t)$ that is actually referred to as the stochastic master equation in the quantum optics literature. The density operator $\hat{\rho}$  can be interpreted as the stochastic density operator of the system as it evolves under continuous observation of $Y(t)$.

Let $\rho_{\rm sys}$ denotes the reduced density operator of the system only, after tracing out the fields to which it is coupled, i.e., $\rho_{\rm red}(t)={\rm Tr}_{\mathcal{F}_n}(U(t) \rho_{\rm sys} |\Omega \rangle \langle \Omega| U(t)^*)$. This master equation can be computed in the standard way via the identity ${\rm Tr}(\rho_{\rm red}(t)X)=\varpi(j_t(X))$ for any $X \in \mathscr{B}(\mathfrak{h}_{\rm sys})$, and the fact that the fields are in a vacuum state; see, e.g., \cite{GJ07}. For the systems considered in this section, we obtain the so-called quantum master equation,
\begin{equation}
\dot{\rho}_{\rm red}(t) = \mathcal{L}_{H,L}^{\star} \rho_{\rm red}(t), \label{eq:master-multi-vac}
\end{equation}
The master equation can also be obtained directly from the quantum filtering equation.  The reduced density operator of the system, $\rho_{\rm red}$, is related to $\hat{\rho}$ via expectation, $\rho_{\rm red}(t)=\varpi(\hat{\rho}(t))$. That is, the reduced density operator of the system at time $t \geq 0$ can be recovered from the  stochastic density operator  $\hat{\rho}$ by averaging the latter over all possible stochastic trajectories induced by continuous measurement of $Y(t)$. By the $\mathscr{Y}_t$-martingale property of the innovation process $\nu$, it is easy to see that computing $d\varpi(\hat{\rho}(t))$ again yields the master equation  (\ref{eq:master-multi-vac}).

\subsection{Quantum filtering for MIMO systems driven by arbitrary zero-mean jointly Gaussian fields}
\label{sec:filtering-Gaussian}
We now finally apply the results of Section \ref{sec:filtering-multi-vac} to systems driven by field in an arbitrary zero-mean jointly Gaussian state under linear measurements made on the system's multiple output fields. Suppose that the system's initial state is $\omega_{\rm sys}(\cdot) ={\rm Tr}(\rho_{\rm sys} \cdot)$ for some density operator $\rho$ in $\mathscr{B}(\mathfrak{h}_{\rm sys})$, and the fields to which it is coupled to is in the jointly Gaussian state $\omega_{N,M}$, with $N,M$ the parameters of the Gaussian state introduced in Section \ref{sec:Gaussian-states}. 
Thus, the joint initial state of the system and fields is $\varpi(\cdot)=\omega_{\rm sys} \otimes \omega_{N,M}$.  Using the generalized Araki-Woods representation (\ref{eq:GAW}) with an appropriate vacuum field state  $|\Omega \rangle$ (i.e., $A_j(f)|\Omega \rangle=0$ for $j=1,2$) on the underlying field Fock state $\mathcal{F}_{2n}$, the state reads 
\begin{eqnarray*}
\varpi(X) = {\rm Tr}( \rho_{\rm sys}\otimes |\Omega\rangle \langle \Omega | X), 
\end{eqnarray*}
for any operator $X$ in $\mathscr{B}(\mathfrak{h}_{\rm sys}) \otimes \mathscr{B}(\mathcal{F}_{2n})$. Consider the vacuum QSDE (\ref{eq:HP-QSDE-nonvac2}) using the generalized Araki-Woods representation of the Gaussian input fields.  The output field is $B^{\rm out}(t)=U(t)^*B(t)U(t)$ and the measurement is of the form 
\begin{equation}
Y(t) = G^{\#} B^{\rm out}(t) + G B^{\rm out}(t)^{\#}, \label{eq:Y-meas-nonvac1} 
\end{equation}
with $[\begin{array}{cc} G^{\#} & G \end{array}]$ assumed to be full rank, as in the vacuum case. Again, using (\ref{eq:GAW}), this can be expressed as 
\begin{equation}
Y(t) = \tilde G^{\#} A^{\rm out}(t) + \tilde G A^{\rm out}(t)^{\#}, \label{eq:Y-meas-nonvac2} 
\end{equation}
for a corresponding $m \times 2n$ matrix $\tilde G$ that is determined by $G$ and the coefficients $C_1,C_2,C_3$ of the generalized Araki-Woods representation. Note that the assumption that  $[\begin{array}{cc} G^{\#} & G \end{array}]$  is full rank implies that $[\begin{array}{cc} \tilde G^{\#} & \tilde G \end{array}]$ must also be full rank. Similarly, we can define $Z$ as
\begin{eqnarray*} 
Z(t) &=& G^{\#} B(t) + G B(t)^{\#},\\
&=& \tilde G^{\#} A(t) + \tilde G A(t)^{\#},
\end{eqnarray*}
with $A(t)=[\begin{array}{cc} A_1(t)^{\top} & A_2(t)^{\top} \end{array}]^{\top}$. Let $\tilde{H} \in \mathbb{C}^{(2n-m)\times 2n}$ and $\tilde{W} \in \mathbb{C}^{2n \times 2n}$ be constructed from $\tilde G$ in the same way that the matrices $H$ and $W$ were constructed from $G$ in Section \ref{sec:filtering-multi-vac}. Moreover,   assume that $\tilde W$ from this construction is invertible. With the definitions that have just been set up, the following proposition follows directly from Proposition \ref{prop:multi-eq-unitary} by setting $S=I$.

\begin{proposition}
Let $U(t)$ be the unitary solution of the QSDE (\ref{eq:HP-QSDE-nonvac2}). Define 
$L_{\tilde{W}} = \tilde{W}^{-\top} L_{N,M}$, and  
$ 
Z_{\tilde W}(t) = \tilde W^{\#} A(t) + \tilde W  A(t)^{\#}. 
$  
Also, let $V(t)$ be an adapted process defined as the solution to the QSDE
$$
dV(t) = \bigl(-\bigl(\imath H+\nicefrac{1}{2}L_{N,M}^*L_{N,M} \bigr)dt +  L_{\tilde W}^{\top} dZ_{\tilde W}(t)\bigr)V(t),
$$
with $V(0)=I$. Then $[Z_{\tilde W}(t),Z(s)]=0$ for all $s,t \geq 0$ and
$$
\varpi(U(t)^* X U(t)) = \varpi(V(t)^* X V(t))
$$
for all $X \in \mathscr{B}(\mathfrak{h}_{\rm sys})$.
\end{proposition}

The following theorem then gives an explicit SDE for the quantum Zakai equation and the quantum Kushner-Stratonovich equation. It  follows from similar calculations as in the proof of Theorems \ref{thm:filter-single-in} and \ref{thm:filter-multi-in}, using the proposition above  and the equivalent vacuum QSDE (\ref{eq:HP-QSDE-nonvac2}) and  equivalent vacuum measurement equation (\ref{eq:Y-meas-nonvac2}).
\begin{theorem}
\label{thm:filter-nonvac-Gaussian} Let $\varpi(Z_{\tilde W}(t) | \mathscr{Z}_t)=K_{\tilde W} Z(t)$ with $K_{\tilde W} \in \mathbb{R}^{2n \times m}$. The quantum Zakai equation for the system (\ref{eq:HP-QSDE-nonvac1}) with measurement $Y$ given by (\ref{eq:Y-meas-nonvac1}) is
\begin{eqnarray}
d\sigma_t(X) = \sigma_t(\mathcal{L}_{H,L_{N,M}}(X)) dt + \sigma_t(L_{\tilde W}^*X+XL_{\tilde W}^{\top})K_{\tilde W} dY(t). \label{eq:Zakai-nonvac}
\end{eqnarray}
and the quantum Kushner-Stratonovich equation for the conditional expectation $\pi_t(X)$ is
\begin{eqnarray*}
d\pi_t(X) &=& \pi_t(\mathcal{L}_{H,L_{N,M}}(X))dt + (\pi_t(L_{\tilde W}^*X+XL_{\tilde W}^{\top}) -\pi_t(X)\pi_t(L_{\tilde W}^*+L_{\tilde W}^{\top})) K_{\tilde W} d\nu(t) \label{eq:KS-nonvac}
\end{eqnarray*}
where $\nu$ is the innovations process and is a $\mathscr{Y}_t$-martingale defined by
$$
\nu(t)=\int_{0}^t (dY(s)-\tilde G^{\#}\tilde G^{\top}K_{\tilde W}^{\top} \pi_s(L_{\tilde W}^{\#}+L_{\tilde W})ds). 
$$  
Moreover, $\nu$ is a Wiener process.
\end{theorem}
  
Also, following Section \ref{sec:filtering-multi-vac}, the stochastic master equation for the stochastic density operator $\hat{\rho}(t)$ is
\begin{equation}
d\hat{\rho}(t) = \mathcal{L}_{H,L_{N,M}}^{\star}(\hat{\rho}(t))dt +\bigl( \hat{\rho}(t) L_{\tilde W}^* +L_{\tilde W}^{\top}\hat{\rho}(t)-{\rm Tr}(\hat{\rho}(t) (L_{\tilde W}^*+L_{\tilde W}^{\top}))\hat{\rho}(t)\bigr)K_{\tilde W} d\nu(t), \label{eq:SME-nonvac}
\end{equation}
and the quantum master equation is 
\begin{equation}
\dot{\rho}_{\rm red}(t) = \mathcal{L}_{H,L_{N,M}}^{\star} \rho_{\rm red}(t). \label{eq:master-nonvac}
\end{equation}

\section{Conclusion}
\label{sec:conclu}
Using the reference probability approach to quantum filtering, in this paper  we have derived the most general form of the quantum filtering equation for Markovian open quantum systems coupled to multiple fields in an arbitrary zero-mean jointly Gaussian state with general linear measurements performed on multiple output fields from the system. This completes, up to a mild assumption relating to the rank of certain matrices relating to the linear measurements, the partial results that are available in the literature for various special classes of zero-mean Gaussian states and under specific kinds of linear measurements.  It is hoped that these general results will be useful for various applications of quantum filtering theory involving Gaussian input states, especially in the context of quantum feedback control and estimation theory.

\bibliographystyle{ieeetran}
\bibliography{ieeeabrv,rip,mjbib2004,irpnew}

\end{document}